\documentclass[a4paper, english, cleveref]{lipics-v2019}
\usepackage{placeins}
\usepackage[nobreak]{mdframed}
\usepackage{xcolor}
\usepackage{thmtools}
\usepackage{thm-restate}
\usepackage[linesnumbered, ruled, titlenotnumbered, noline]{algorithm2e}
\hypersetup{colorlinks = true, citecolor = blue, linkcolor = blue}

\definecolor{blue}{rgb}{0.032812499999999994, 0.3390625, 0.45390624999999996}

\nolinenumbers
\usepackage{amsmath, amsthm, bm, amssymb}

\crefname{theorem}{Theorem}{Theorems}
\crefname{lemma}{Lemma}{Lemmas}
\crefname{corollary}{Corollary}{Corollaries}
\crefname{definition}{Definition}{Definitions}
\crefname{claim}{Claim}{Claims}
\crefname{remark}{Remark}{Remarks}

\newcommand{\OPT}{\mathit{OPT}}
\newcommand{\ALG}{\mathit{ALG}}
\newcommand{\E}{\mathcal{E}}
\newcommand{\X}{\mathcal{X}}
\newcommand{\A}{\mathcal{A}}
\newcommand{\D}{\mathcal{D}}
\newcommand{\G}{\mathcal{G}}

\newcommand{\Outer}{\operatorname{\textsc{outer}}}
\newcommand{\problem}[4]{
  \begin{mdframed}%
  \begin{tabularx}{\textwidth}{l X}%
  \multicolumn{2}{l}{#1 #2} \\%
  \textbf{Input:} & #3\\%
  \textbf{Question:} & #4%
  \end{tabularx}%
  \end{mdframed}%
}

\newcommand{\sat}{\textsc{3-sat}}
\newcommand{\psat}{\textsc{Planar-3-sat}}
\newcommand{\tracking}{\textsc{Tracking}}
\newcommand{\ptracking}{\textsc{Planar-Tracking}}

\title{Tracking Paths in Planar Graphs}
 
\author{David Eppstein}
  {Department of Computer Science, University of California Irvine, US}
  {eppstein@uci.edu}{}{}
\author{Michael T. Goodrich}
  {Department of Computer Science, University of California Irvine, US}
  {goodrich@uci.edu}
  {https://orcid.org/0000-0002-8943-191X}{}
\author{James A. Liu}
  {Department of Computer Science, University of California Irvine, US}
  {jamesal1@uci.edu}{}{}
\author{Pedro Matias}
  {Department of Computer Science, University of California Irvine, US}
  {pmatias@uci.edu}
  {https://orcid.org/0000-0003-0664-9145}{}

\authorrunning{D.\,Eppstein, M.\,T. Goodrich, J.\,A. Liu and P.\,Matias}

\Copyright{David Eppstein, Michael T. Goodrich, James A. Liu and Pedro Matias}



\acknowledgements{We thank Nil Mamano for suggesting the problem of tracking paths on a graph.}

\EventEditors{Pinyan Lu and Guochuan Zhang}
\EventNoEds{2}
\EventLongTitle{30th International Symposium on Algorithms and Computation (ISAAC 2019)}
\EventShortTitle{ISAAC 2019}
\EventAcronym{ISAAC}
\EventYear{2019}
\EventDate{December 8--11, 2019}
\EventLocation{Shanghai University of Finance and Economics, Shanghai, China}
\EventLogo{}
\SeriesVolume{149}
\ArticleNo{57}

\begin{document}

\maketitle

\begin{abstract}\label{abs}
We consider the NP-complete problem of tracking paths in a graph, first introduced by Banik et. al. \cite{banik2017tracking}. Given an undirected graph with a source $s$ and a destination $t$, find the smallest subset of vertices whose intersection with any $s-t$ path results in a unique sequence. In this paper, we show that this problem remains NP-complete when the graph is planar and we give a 4-approximation algorithm in this setting. We also show, via Courcelle's theorem, that it can be solved in linear time for graphs of bounded-clique width, when its clique decomposition is given in advance.
\end{abstract}

\section{Introduction}\label{sec:intro}

Motivated by applications in surveillance and monitoring, Banik et. al. \cite{banik2017tracking,banik2018polynomial} introduced the problem of tracking paths in a graph. In essence, the goal is to uniquely determine the path traversed by a moving subject or object, based on a sequence of vertices sampled from that path. Examples of surveillance applications include the following: (i) vehicle tracking in road networks; (ii) habitat monitoring; (iii) intruder tracking and securing large infrastructures; (iv) tracing back of illicit Internet activities by tracking data packets. Another application would be to determine the nodes in a network that have been compromised by a spreading infection, given an incomplete transmission history of a pathogen. This information can be helpful in identifying and attenuating the negative impact caused by biological and non-biological infectious agents, such as:

\begin{itemize}
  \item Highly-contagious diseases (e.g. Severe Acute Respiratory Syndrome) which could lead to epidemics \cite{goh2006epidemiology}.
  \item Fake news and hate speech being disseminated in social networks, as well as violations of privacy (e.g. sharing without permission highly sensitive content owned by a user, such as intimate pictures).
  \item Computer viruses, which spread throughout servers scattered across the Internet.
\end{itemize}

Some of these applications have been studied empirically or using heuristics in \cite{bhatti2009survey,gupta2003tracking,peng2007survey,snoeren2001hash} or, in the case of infections, \cite{moore2000epidemics,newman2002spread,shah2011rumors,bailey1975mathematical}. To the best of our knowledge, Banik et. al. were the first to approach the problem of tracking paths from a theory perspective. In this work, we extend some of their work and give new algorithms. Our main results apply to graphs that can: be embedded on the plane (of interest to surveillance in road networks and similar infrastructures) or that have bounded clique-width (mainly of theoretical interest).

\subparagraph*{Preliminaries.}

In the tracking paths problem, we are given an undirected graph $G=(V,E)$ with no self loops or parallel edges and a source $s \in V$ and a destination $t \in V$. The goal is to place trackers on a subset of the vertices in a way that enables us to reconstruct exactly the path traversed from $s$ to $t$. Let $T \subseteq V$ be a set of vertices (where we wish to place trackers) and let $\mathcal{S}_P^T$ be the \emph{sequence} of vertices in $T$ visited during the traversal of a path\footnote{Some authors use the terms ``path'' and ``walk'' interchangeably, where vertices may be repeated, but in this paper, paths are required to have distinct vertices.} $P$. Let $u-v$ denote a path from $u$ to $v$. We say that $T$ is a \emph{tracking set} if every $s-t$ path yields a unique sequence of observed vertices in $T$, that is, $\mathcal{S}^T_{P_1} \neq \mathcal{S}^T_{P_2}$ for all distinct $s-t$ paths $P_1$ and $P_2$. We consider the following problem.

\problem{\ptracking}{$(G,s,t)$}
  {Undirected \emph{planar} graph $G=(V,E)$ and two vertices $s\in V$ and $t\in V$.}
  {What is the smallest tracking set for $G$?}

We denote by \tracking{} the problem of tracking paths when the input graph is not restricted to be planar. Due to space constraints, we defer proofs of Lemmas/Theorems marked with $\star$ to the appendix.

\subparagraph*{Related work.}

Banik et. al. first introduced \tracking{} in \cite{banik2018polynomial}, where it is shown to be  NP-hard by reducing from Vertex Cover, which seems unlikely to work in the planar case. Although not immediately obvious, they also show that \tracking{} is in NP, by observing that every tracking set is also a \emph{feedback vertex set}, i.e. a set of vertices whose removal yields an acyclic graph. Finally, they present a fixed-parameter tractable (FPT) algorithm (parameterized by the solution size) for the decision version of the problem, where they obtain a kernel of size $O(k^7)$ edges.

The concept of tracking set, however, first appeared in Banik et. al. \cite{banik2017tracking}, where they considered a variant of \tracking{} that only concerns shortest $s-t$ paths, essentially modeling the input as a directed acyclic graph (DAG). Using a similar reduction from Vertex Cover, they show that this variant cannot be approximated within a factor of $1.3606$, unless P=NP. They also give a 2-approximation for the planar version of tracking shortest paths, but they omit any hardness results for this variant. 

More recently, Bil\`{o} et. al. \cite{DBLP:conf/sirocco/BiloG0P19} generalized the version of the problem concerning shortest $s-t$ paths, into the case of multiple source-destination pairs, for which they claim the first $O(\sqrt{n \log n})$-approximation algorithm for general graphs. They also study a version of this multiple source-destination pairs problem in which the set of trackers itself (excluding the order in which they are visited) is enough to distinguish between $s-t$ shortest paths\footnote{Notice that when tracking shortest paths only and using a single source-destination pair, these two versions of the problem are the same.}. In this setting, they claim a $O(\sqrt{n})$-approximation algorithm and they show that it is NP-hard even for cubic planar graphs. The hardness construction intrinsically relies on the multiplicity of source-destination pairs and, therefore, cannot be adapted to the problem studied in this paper. They also give an FPT algorithm (w.r.t to the maximum number of vertices at the same distance from the source) for the problem concerning a single source-destination pair that was introduced in \cite{banik2017tracking}.

In \cite{DBLP:conf/caldam/BanikC18}, Banik and Choudhary generalize \tracking{} into a problem on set systems\footnote{Also called hypergraphs.}, which are characterized by a universe (e.g. vertex set) and a family of subsets of the universe (e.g. $s-t$ paths). They show that this generalized version of the problem is fixed-parameter tractable, by establishing a correspondence with the well known Test Cover problem.

\subparagraph*{Our results.}

In this paper, we give a 4-approximation for \ptracking{} (\cref{sec:approximation}) and prove that it is NP-complete (\cref{sec:hardness}). In addition, we show that \tracking{} can be solved in cubic time for graphs of bounded clique-width and linear time if the clique decomposition of bounded width is given in advance (\cref{sec:courcelle}).

\section{Definitions}

\begin{definition}[Entry-exit pair] \label{def:entry_exit}
Let $(G,s,t)$ be an instance of \tracking{}. An \emph{entry-exit} pair is, with respect to some simple cycle $C$ in $G=(V,E)$, an ordered pair $(s',t')$ of vertices in $C$ that satisfy the following conditions:
\begin{enumerate}
  \item There exists a path $s-s'$ from $s$ to the \emph{entry} vertex $s'$
  \item There exists a path $t'-t$ from the \emph{exit} vertex $t'$ to $t$
  \item Paths $s-s'$ and $t'-t$ are vertex-disjoint
  \item Path $s-s'$ (resp. $t'-t$) and $C$ share exactly one vertex: $s'$ (resp. $t'$).
\end{enumerate}
\end{definition}

Essentially, an entry-exit pair $(s',t')$ with respect to a cycle $C$ (see \cref{fig:entry_exit}) represents two alternative $s-t$ paths and, thus, requires tracking at least one of them. We say that $(s',t')$ is \emph{tracked with respect to $C$} if and only if $C\setminus \{s',t'\}$ contains a tracker. In addition, $C$ is \emph{tracked} if and only if there is no entry-exit pair with respect to $C$ that is untracked. If a cycle contains either (i) 3 trackers or (ii) $s$ or $t$ and 1 tracker in a non-entry/non-exit vertex, then it must be tracked. We say that these cycles are \emph{trivially tracked}.

\begin{figure}
\centering
\includegraphics[width=190pt]{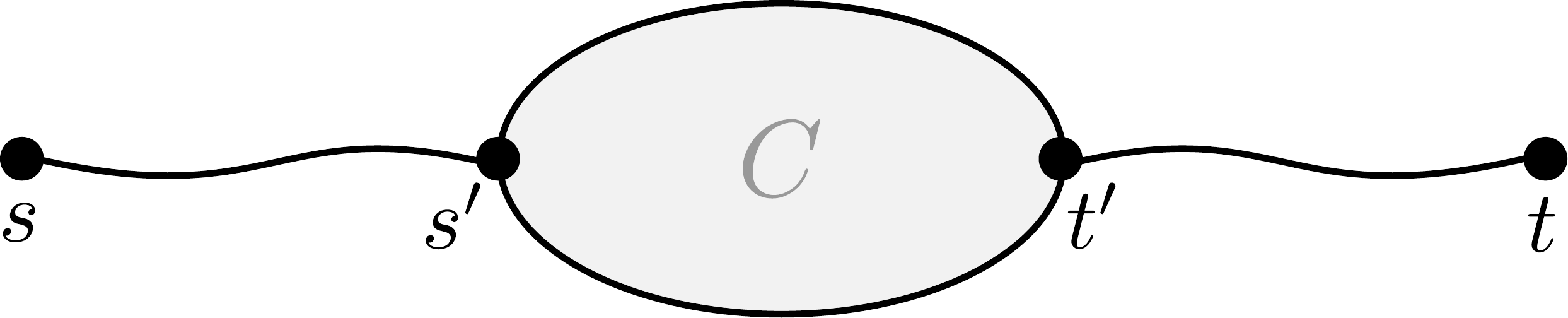}
\caption{Entry-exit pair illustration, with entry vertex $s'$ and exit vertex $t'$.}
\label{fig:entry_exit}
\end{figure}

An alternative characterization of a tracking set, first given by Banik et al \cite[Lemma~2]{banik2018polynomial}, is the following.

\begin{lemma}[\cite{banik2018polynomial}] \label{lem:tracking_set}
For a graph $G=(V,E)$, a subset $T \subseteq V$ is a tracking set if and only if every simple cycle $C$ in $G$ is tracked with respect to $T$.
\end{lemma}

\section{Approximation algorithm}\label{sec:approximation}

\begin{theorem} \label{thm:4approx}
There exists a 4-approximation algorithm for \ptracking{}.
\end{theorem}

The overall idea for the approximation algorithm builds on the following two insights:
\begin{enumerate}
  \item The cardinality of the optimal solution cannot be much smaller than the number of faces in the graph.
  \item The average number of trackers per face does not need to be very large.
\end{enumerate}

The first idea gives us a lower bound on $\OPT$, the cardinality of an optimal solution. The second gives us an upper bound on $\ALG$, the cardinality of our approximation algorithm.

\subsection{Lower bound on \texorpdfstring{$\OPT$}{OPT}}

We consider the following reduction, which takes care of disconnected components, or components that are ``attached'' to the graph by a cut vertex that is not in an $s-t$ path. We say that a reduction is \emph{safe}, if it does not eliminate any untracked cycles.

\begin{enumerate}[{Reduction} 1.]
\item While there exists an edge or vertex that does not participate in any $s-t$ path, remove it from the graph. \label{reduction:1}
\end{enumerate}

\begin{restatable}[\cite{banik2018polynomial}]{lemma}{lemrone}\label{lem:reduction1}
$\star$ \hyperref[reduction:1]{Reduction~1} is safe and can be done in polynomial time.
\end{restatable}

\begin{lemma}\label{lem:cycle_has_entry_exit}
After \hyperref[reduction:1]{Reduction~1}, every simple cycle in the graph contains at least one entry-exit pair. This holds for non-planar graphs as well.
\end{lemma}

\begin{proof}
Let $C$ by some simple cycle in the graph. After \hyperref[reduction:1]{Reduction~1}, there must exist an $s-t$ path that shares an edge with $C$. The first and last vertices on this path that belong to $C$ correspond to an entry-exit pair.
\end{proof}

\begin{lemma}\label{lem:OPT_bound}
In an embedded undirected planar graph $G$ that results from \hyperref[reduction:1]{Reduction~1}, $\OPT \ge (|F|-1)/2$, where $F$ is the set of faces of $G$.
\end{lemma}

On a high level, the proof of \cref{lem:OPT_bound} is done by keeping a set of "active" trackers while reconstructing a planar embedding $\E$ of $G$: we start, as a base case, with any simple $s-t$ path in $\E$ and iteratively add faces to it until it matches $\E$. Given a fixed, optimal tracking set $T^*$, the addition of each face requires either (i) adding a new tracker from $T^*$ to the active set, or (ii) deactivating an active tracker, rendering it useless for distinguishing paths on future faces. As a consequence, each tracker charges at most two faces: the one adding the tracker and the one deactivating it. This demonstrates that $|T^*|\geq (|F|-1)/2$.

Let $\Outer(\E^\tau)$ be the set of outer-edges of our planar reconstructing embedding $\E^\tau$ at time $\tau$. At time $\tau=0$, our embedding corresponds to an $s-t$ path and, for all $0 \le \tau \le |F|-1$, we add exactly one face $C$ from $\E$ to $\E^{\tau}$, by connecting two vertices $u$ and $v$ in $\Outer(\E^{\tau})$ with a simple path $p$ (see \cref{fig:add_face} in the appendix). In doing so, we erase an $u-v$ path $p'$ in $\Outer(\E^{\tau})$, so we have that $\Outer(\E^{\tau+1}) = \Outer(\E^{\tau}) \setminus p' \cup p$. In the end, $\E^{|F|} = \E$.

By \cref{lem:cycle_has_entry_exit}, there is at least one entry-exit pair in $\E$ with respect to face $C$, so any tracking set must contain a tracker on some vertex of $C$. During the reconstruction process, we maintain a list of trackers in sets $A$ and $A'$, such that $(A\cup A')\subseteq T^*$, where $A$ contains active trackers and $A'$ contains inactive ones. A tracker in vertex $v$ is \emph{active} at time $\tau$ if and only if it meets both of the following conditions:

\begin{enumerate}[{Condition} (i)]
  \item $v \in \Outer(\E^{\tau})$ \label{condition:c1}
  \item There is no $s-v$ path in $\E^{\tau}$ that traverses vertices in $\Outer(\E^{\tau})$
  \label{condition:c2}
\end{enumerate}

Intuitively, an active tracker can be used to track future faces, although no more than one (see below). An inactive tracker, on the other hand, either cannot be used to track future faces (\hyperref[condition:c1]{Condition~(i)}), or its corresponding vertex is entry/exit for some future face (\hyperref[condition:c2]{Condition~(ii)}), in which case we require yet another tracker on that face (see \cref{fig:inactive_trackers}). \hyperref[condition:c2]{Condition~(ii)} is necessary for dealing with embeddings of $G$ where at least one of $s$ and $t$ is not in the outer face.

\begin{figure}
\centering
\begin{subfigure}[t]{0.45\textwidth}\label{subfig:inactive_tracker1}
\centering
  \includegraphics[height=90pt]{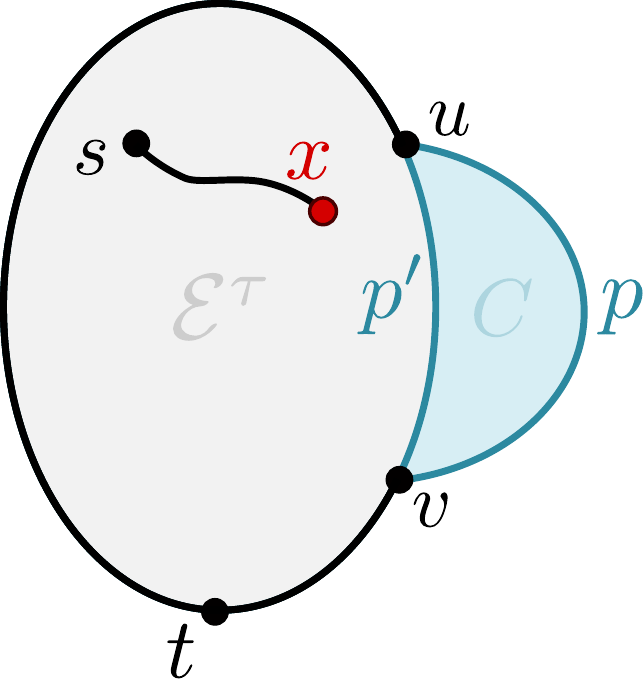}
  \caption{Tracker on $x$ (red) is inactive due to the violation of \hyperref[condition:c1]{Condition~(i)}.}
\end{subfigure}
\hspace{2em}
\begin{subfigure}[t]{0.45\textwidth}\label{subfig:inactive_tracker2}
\centering
  \includegraphics[height=90pt]{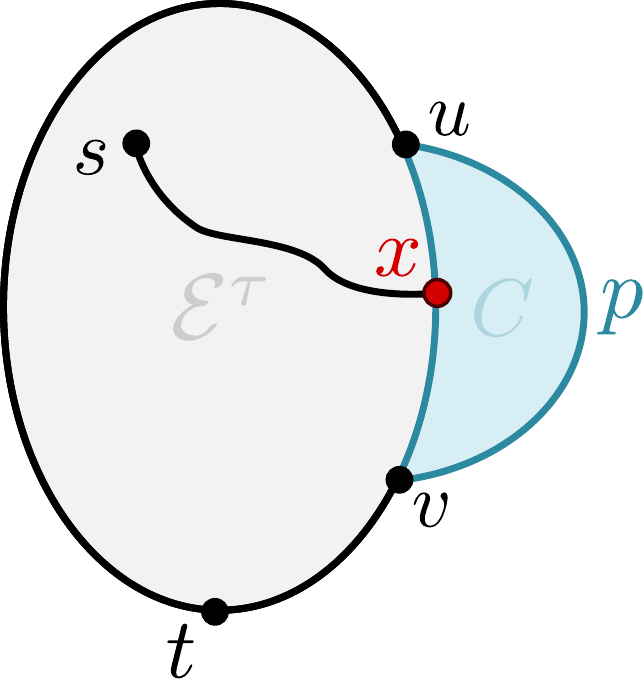}
  \caption{Tracker on $x$ (red) is inactive due to the violation of \hyperref[condition:c2]{Condition~(ii)}. Notice that $x$ is entry for exit vertices $u$ and $v$ (with respect to $C$), therefore $C$ needs another tracker.}
\end{subfigure}
\caption{Examples of inactive trackers used in the proof of \cref{lem:OPT_bound}.}
\label{fig:inactive_trackers}
\end{figure}

\begin{proof}[Proof of \cref{lem:OPT_bound}]

~

First, we argue that each time we add a face during the reconstruction process described above, we either (i) need to increase the number of active trackers (by adding it to either $A$ or $A'$), or (ii) we can get away by re-using and, therefore, deactivating an active tracker.

We assume for the rest of the argument that $t$ is on the outer face, because such a planar embedding is always possible to construct.

Let $C$ be the face added a time $\tau$ by connecting vertices $u$ and $v$ in $\Outer(\E^{\tau})$, as specified above. In addition, let $T^*$ be any optimal solution (i.e. $|T^*|= \OPT$). We consider two cases, depending on the existence of a tracker in $C$ at time $\tau$:

\begin{enumerate}[{Case} 1:]
  \item $\bm{C\cap A = \emptyset}$ at time $\tau$. \label{case:C_has_no_tracker}

  By \cref{lem:cycle_has_entry_exit}, there exists a vertex $x \in C$ such that $x\in T^*$.  We place a tracker on $x$. If $x\in \Outer (\E^ {\tau+1})$ we add $x$ to $A$, otherwise we add it to $A'$.

  \item $\bm{C\cap A \neq \emptyset}$ at time $\tau$. \label{case:C_has_tracker}

  Let $y \in C\cap A$ be a vertex of $C$ with a tracker. We again consider two cases:
  \begin{enumerate}[(i)]
    \item $\bm{y \notin \{u,v\}}.$ Then, $y \in \Outer(\E^\tau)$ but $y \notin \Outer(\E^{\tau+1})$, which amounts to moving $y$ from $A$ to $A'$.

    \item $\bm{y \in \{u,v\}}.$ If $(u,v)$ is an entry-exit pair with respect to $C$, or if the tracker in $y$ is not active, then there exists $x' \in C\setminus \{u,v\}$ such that $x' \in T^*$. Similarly to \hyperref[case:C_has_no_tracker]{Case~1}, we place a tracker on $x'$, which corresponds to adding $x'$ either to $A$ or $A'$.

    Otherwise, the tracker in $y$ is active and $(u,v)$ is not an entry-exit pair with respect to
    $C$. Let us assume without loss of generality that $y=u$. Then, the addition of $C$ deactivates the tracker in $u$ by definition of active tracker (\hyperref[condition:c2]{Condition~(ii)} is now violated), so we move $u$ from $A$ to $A'$.
  \end{enumerate}

  Every tracker in $A'$ is charged by at most two faces: one for adding an active tracker to $A$ and another for deactivating it and moving it to $A'$. Therefore, $|F| - 1 \le |A| + 2|A'|$. Since $|A| + |A'| \le |T^*|$, it follows that $|F| -1 \le 2\OPT$.
\end{enumerate}
\end{proof}

A tight example for the lower bound on $OPT$ is illustrated in \cref{subfig:OPT_tight}.

\begin{figure}
\centering
\begin{subfigure}[b]{0.4\textwidth}
\centering
\includegraphics[height=100pt]{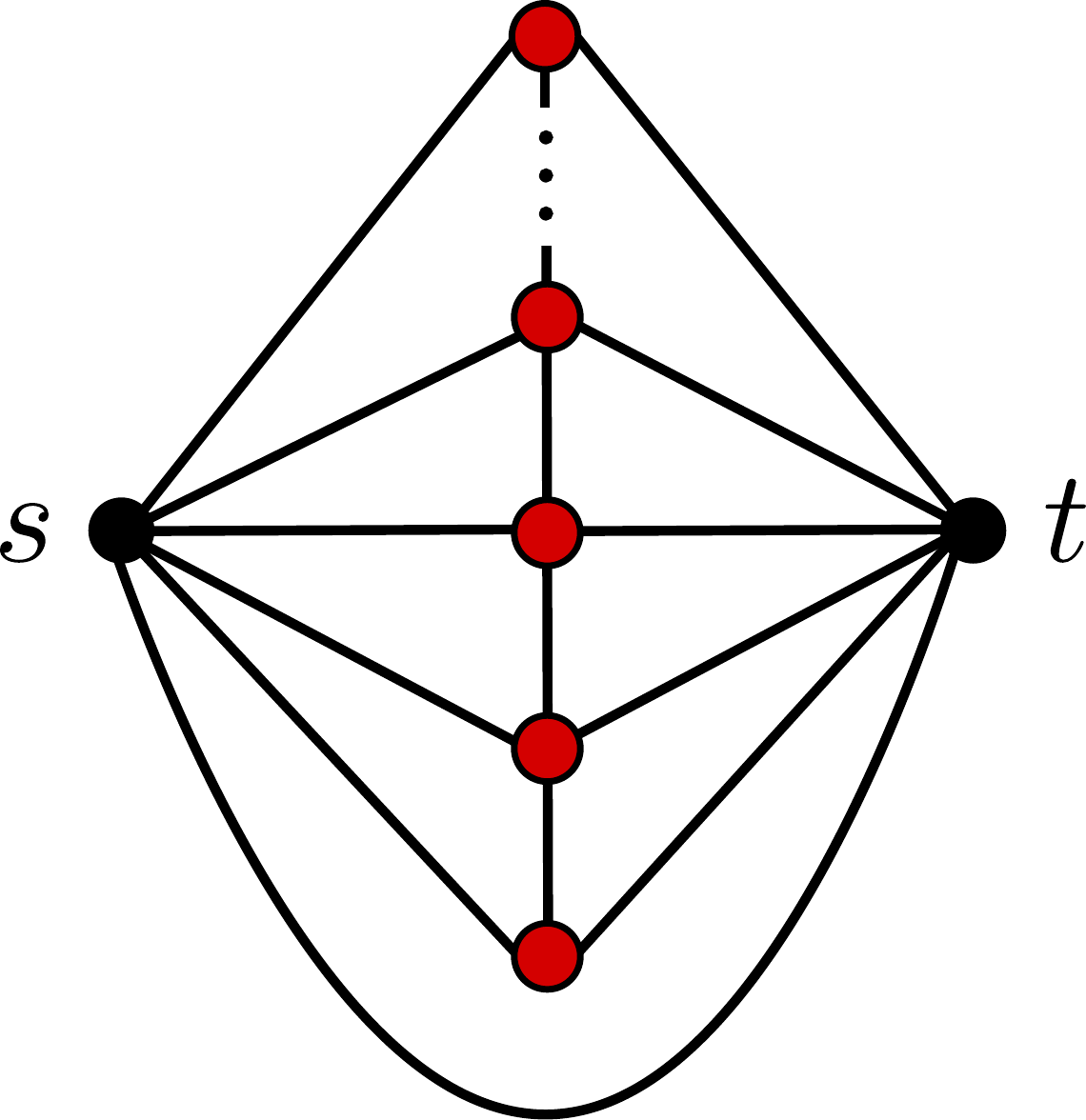}
\caption{Example of a planar graph where $OPT = |F|/2$ (in red).}\label{subfig:OPT_tight}
\end{subfigure}
\hspace{2em}
\begin{subfigure}[b]{0.5\textwidth}
\centering
\includegraphics[width=130pt]{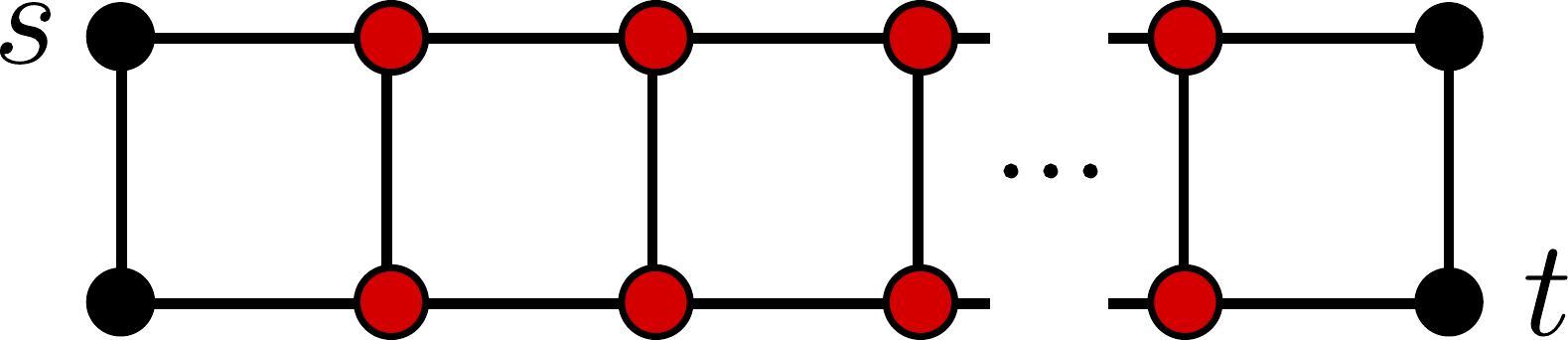}
\caption{Example of a planar graph where $ALG=2(|F|-2)$ (in red).}\label{subfig:ALG_tight}
\end{subfigure}
\caption{Tight examples for the lower bound on $OPT$ (left) and the upper bound on $ALG$ (right) in planar graphs.}
\end{figure}

\subsection{Upper bound on \texorpdfstring{$\ALG$}{ALG}}

We say that an undirected planar graph is \emph{reduced} if it cannot be further reduced by \hyperref[reduction:1]{Reduction~1} or any of the following reductions. 

\begin{enumerate}[{Reduction} 1.]
\setcounter{enumi}{1}
\item While there exist two adjacent vertices of degree 2, remove one of them (and its edges) and add an edge connecting its neighbors. \label{reduction:2}

\item While there exists vertex $v \notin \{s,t\}$ of degree 2 in a 3-cycle, place a tracker on $v$ and remove it and its edges from the graph. \label{reduction:3}

\item While there exist non-adjacent vertices $u,v \notin \{s,t\}$ of degree 2 in a 4-cycle, place a tracker on either $u$ or $v$ and remove it and its edges from the graph. \label{reduction:4}
\end{enumerate}

\begin{figure}[b]
\centering
\begin{subfigure}[c]{.45\textwidth}
  \begin{subfigure}[t]{\textwidth}
  \centering
    \includegraphics[width=178pt]{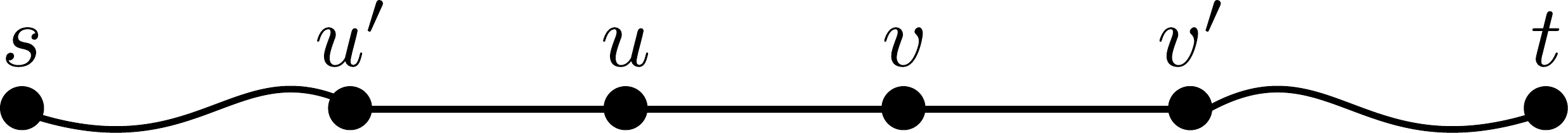}
  \caption{Illustration of \hyperref[reduction:2]{Reduction~2}, where $deg(u) = deg(v) = 2$.}
  \label{fig:reduction2}
  \end{subfigure}

  \vspace{1em}

  \begin{subfigure}[t]{\textwidth}
  \centering
    \includegraphics[width=130pt]{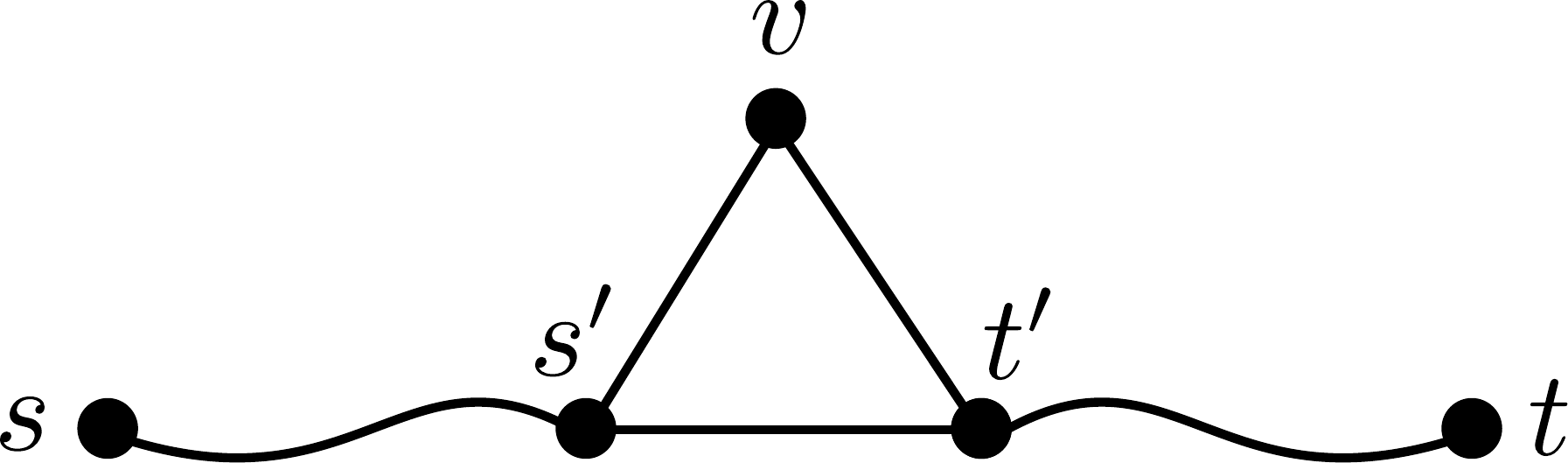}
  \caption{Illustration of \hyperref[reduction:3]{Reduction~3}, where $deg(s')\ge 3$, $deg(t')\ge 3$ and $deg(v)=2$.}
  \label{fig:reduction3}
  \end{subfigure}
\end{subfigure}
\hspace{2em}
\begin{subfigure}[c]{.45\textwidth}
\centering
\includegraphics[height=100pt]{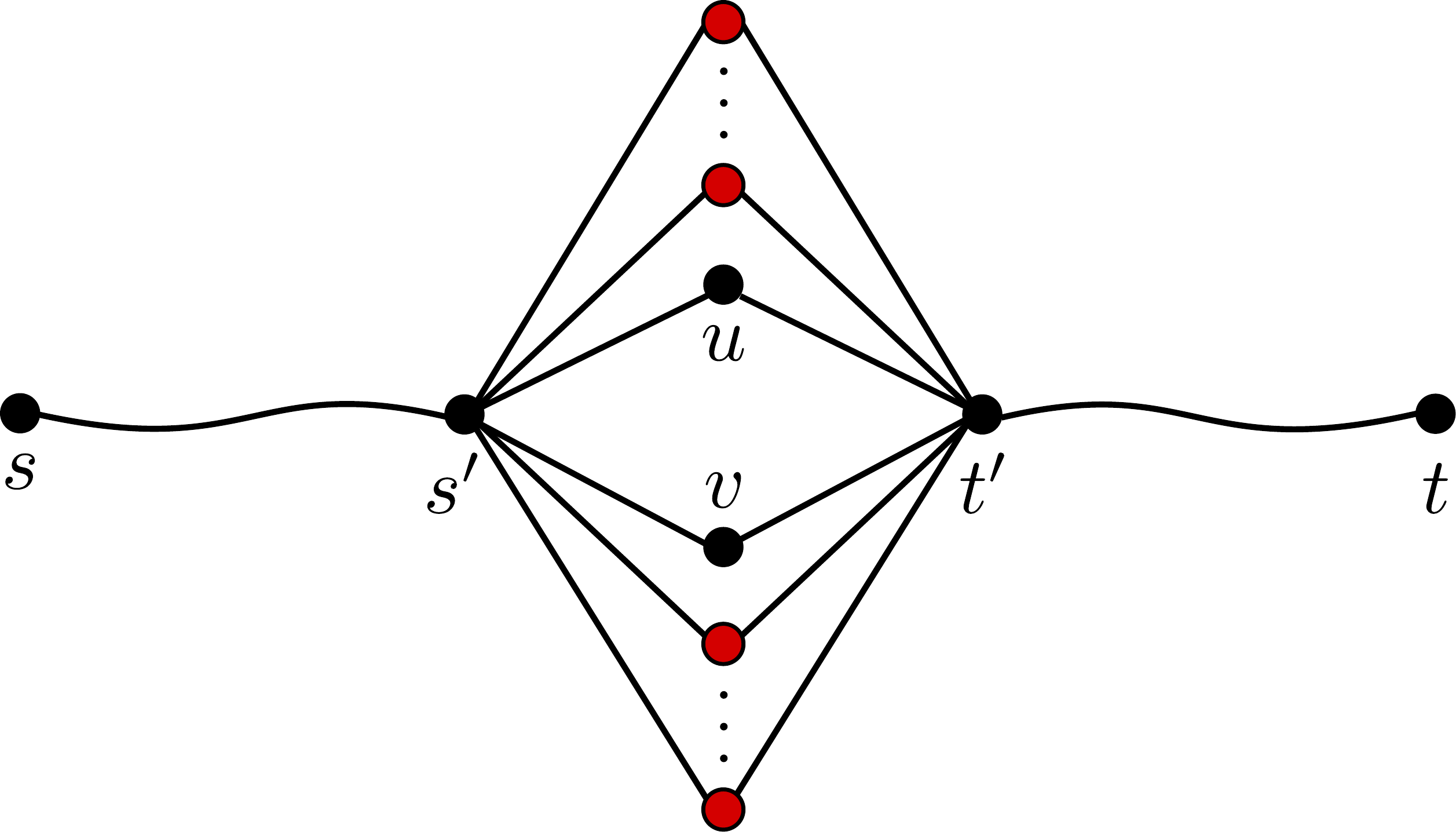}
\caption{Illustration of \hyperref[reduction:4]{Reduction~4}, where $deg(s')\ge 3$, $deg(t')\ge 3$ and $deg(u)=deg(v)=2$.}
\label{fig:reduction4}
\end{subfigure}
\caption{Illustration of \hyperref[reduction:2]{Reductions~2}, \hyperref[reduction:3]{3} and \hyperref[reduction:4]{4}.}
\end{figure}

Notice that all reduction rules are valid for general graphs, not only planar ones. In addition, \hyperref[reduction:2]{Reductions~2}, \hyperref[reduction:3]{3} and \hyperref[reduction:4]{4} can be applied interchangeably and in any order until none of them is applicable, but we will see that they need to be carried out after \hyperref[reduction:1]{Reduction~1}. Fortunately, we will not be required to re-apply \hyperref[reduction:1]{Reduction~1} after performing the remaining reductions.

We denote the degree of a vertex $v$ by $deg(v)$, where the underlying graph can be determined from its context.

\begin{claim}\label{clm:deg2_cant_entry_exit}
If vertex $v$ is on an entry-exit pair, then $deg(v) > 2$.
\end{claim}

\begin{proof}
Trivial by \cref{def:entry_exit}.
\end{proof}

\begin{claim}
\hyperref[reduction:2]{Reductions~2}, \hyperref[reduction:3]{3} and \hyperref[reduction:4]{4} maintain the property that every cycle in $G$ contains at least one entry-exit pair (see \cref{lem:cycle_has_entry_exit}).
\end{claim}

\begin{proof}
\hyperref[reduction:2]{Reductions~2}, \hyperref[reduction:3]{3} and \hyperref[reduction:4]{4} only erase faces and vertices of degree 2, which cannot be in entry-exit pairs (by \cref{clm:deg2_cant_entry_exit}), so every simple cycle of the graph still contains an entry-exit pair.
\end{proof}

\begin{restatable}{lemma}{lemrtwo}\label{lem:reduction2}
$\star$ \hyperref[reduction:2]{Reduction~2} is safe and can be done in polynomial time, if done after \hyperref[reduction:1]{Reduction~1}.
\end{restatable}

\begin{restatable}{lemma}{lemrthree}\label{lem:reduction3}
$\star$ \hyperref[reduction:3]{Reduction~3} is safe and can be done in polynomial time, if done after \hyperref[reduction:1]{Reduction~1}.
\end{restatable}

\begin{restatable}{lemma}{lemrfour}\label{lem:reduction4}
$\star$ \hyperref[reduction:4]{Reduction~4} is safe and can be done in polynomial time, if done after \hyperref[reduction:1]{Reduction~1}.
\end{restatable}

\begin{remark}
None of \hyperref[reduction:2]{Reductions~2}, \hyperref[reduction:3]{3} and \hyperref[reduction:4]{4} compromise planarity.  
\end{remark}

\begin{algorithm}
\SetKwInOut{Input}{Input}\SetKwInOut{Output}{Output}
\Input{Undirected planar graph $G=(V,E)$ and vertices $s\in V$ and $t\in V$}
\Output{Tracking set}

\BlankLine
Perform \hyperref[reduction:1]{Reduction~1} in $G$

Perform \hyperref[reduction:2]{Reductions~2}, \hyperref[reduction:3]{3} and \hyperref[reduction:4]{4} repeatedly until $G$ is reduced.

Output remaining vertices of degree at least 3 (except $s$ or $t$)
\caption{$\mathcal{A}$}\label{alg:planar}
\end{algorithm}

\begin{lemma}\label{lem:A_trackingset}
  Algorithm $\A$ outputs a tracking set for the input graph $G$.
\end{lemma}

\begin{proof}
By \cref{lem:reduction1,,lem:reduction2,,lem:reduction3,lem:reduction4}, \hyperref[reduction:1]{Reductions~1}\hyperref[reduction:4]{-4} are safe, so let us assume without loss of generality that $G$ is reduced. Then, every cycle of $G$ of 5 or more vertices is trivially tracked, because it must contain at least 3 vertices of degree at least 3 (by \hyperref[reduction:2]{Reduction~2}). Similarly, every 3- or 4-cycle must contain at least 3 vertices of degree at least 3 by \hyperref[reduction:2]{Reduction~2} and \hyperref[reduction:3]{Reductions~3} and \hyperref[reduction:4]{4} (respectively).
\end{proof}

\begin{lemma}\label{lem:A_bound}
  Algorithm $\A$ outputs a tracking set of size at most $2(|F|-2)$, where $F$ is the set of faces of the input graph $G$.
\end{lemma}

\begin{proof}
Notice that each tracker added during \hyperref[reduction:3]{Reductions~3} and \hyperref[reduction:4]{4} is associated with the removal of one face from $G$. Therefore, it is enough to show that the lemma holds with respect to a reduced graph $G$. Let us partition $V$ into $V=(V_{2} \cup V_{\ge 3})$, where $V_2$ and $V_{\ge 3}$ consist of the vertices of degree 2 and degree at least 3, respectively (notice that there cannot be vertices of degree 1). The lemma statement follows from \cref{lem:A_trackingset} and the fact that $|F| \ge \frac{|V_{\ge 3}|}{2} + 2$. This inequality can be derived by plugging in the inequality $2|E| \ge 3|V_{\ge 3}| + 2|V_2|$ in Euler's formula for planar graphs: $|V|-|E|+|F|= 2$, where $E$ is the set of edges of $G$.
\end{proof}

A tight example is illustrated in \cref{subfig:ALG_tight}.

\begin{mdframed}
\begin{proof}[Proof of \cref{thm:4approx}]
  By \cref{lem:A_bound,lem:OPT_bound}, Algorithm $\A$ is a 4-approximation to \ptracking{}.
\end{proof}
\end{mdframed}

\section{Hardness of tracking paths}\label{sec:hardness}

We show that \ptracking{} is NP-hard, by reducing from \psat{}, a special version of the satisfiability problem, shown to be NP-complete by Lichtenstein \cite{lichtenstein1982planar}.

In \sat{}, we are given a set $\X=\{x_1,\dots,x_p\}$ of variables and a 3-CNF formula $\phi$, where each clause in $\phi$ is a disjunction of exactly three distinct literals with respect to $\X$. The goal is to find a boolean assignment to all variables in $\X$ that satisfies $\phi$. Consider the bipartite graph with a vertex for each clause $C$ in $\phi$ and each variable $x_i\in \X$, and edges $(x_i, C)$ if and only if $C$ contains $x_i$ or its negation $\overline{x_i}$. Lichtenstein \cite{lichtenstein1982planar} showed that \psat{}, the subset of instances of \sat{} whose underlying bipartite graph is planar, remains NP-complete. In particular, the definition of \psat{} requires that a cycle can be drawn connecting all of the variables while maintaining planarity. Later, Knuth and Raghunatan \cite{knuth1992problem} exploited this condition to show that we can always draw the underlying bipartite graph of a \psat{} instance in a \emph{rectilinear} fashion without crossings (example in \cref {subfig:rect_sat}): variables are arranged in a horizontal line and clauses are horizontal line segments with vertical legs to represent the literals present in the clause. Vertical legs attach to the appropriate variables and are labeled \emph{red} for negated literals and \emph{blue}, otherwise. In particular, a given clause is drawn completely above or below the line of variables.

We convert a planar rectilinear drawing $\D$ of an instance of \psat{}, with formula $\phi$ and a set $\X$ of variables, into a planar drawing $\G$ corresponding to the instance of \ptracking{}. The reduction is straightforward: 

\begin{enumerate}[(i)]
  \item transform each variable $x_i$ in $\D$ into a gadget containing $m_i$ copies of literal vertices $x_i$ and $\overline {x_i}$;
  \item transform each 3-legged clause into a face containing corresponding literals vertices and an entry-exit pair;
  \item choose the boolean assignment according to the placement of trackers, such that a clause is satisfied if and only if its corresponding face is tracked.
\end{enumerate}

The union of all the variable and clause gadgets constitutes $\G$ (see example in \cref{fig:reduction_overview} in the appendix). Details of each gadget are given below. For simplicity, we avoid introducing too many subscripts and we rely on pictures to describe the gadgets.

\subparagraph*{Variable gadget.}\label{par:var_gadget}

Each variable gadget converts a variable $x_i$ in $\D$ into a connected subgraph corresponding to \cref{fig:var_gadget}, with length parameterized by $m_i$. We refer to the set $\{h_k,\mu_k,l_k\}$ as column $k$ and we refer to the vertices $\{h_1,\dots,h_{m_i}\} \cup \{l_1,\dots,l_{m_i}\}$ as \emph{literal vertices}.

Each variable gadget is linked with the next one by setting $t_i=s_{i+1}$, to form a horizontal chain of gadgets, where $s=s_1$ and $t=t_p$. For convenience, we force trackers in all the $s_i$ (except $s$), by drawing edges between the $\alpha'$ ($\beta'$) of a variable gadget and the $\alpha$ ($\beta$) of the next variable gadget in the chain. We refer to the resulting drawing as the \emph{spine}.

\begin{figure}
\centering
\includegraphics[width=280pt]{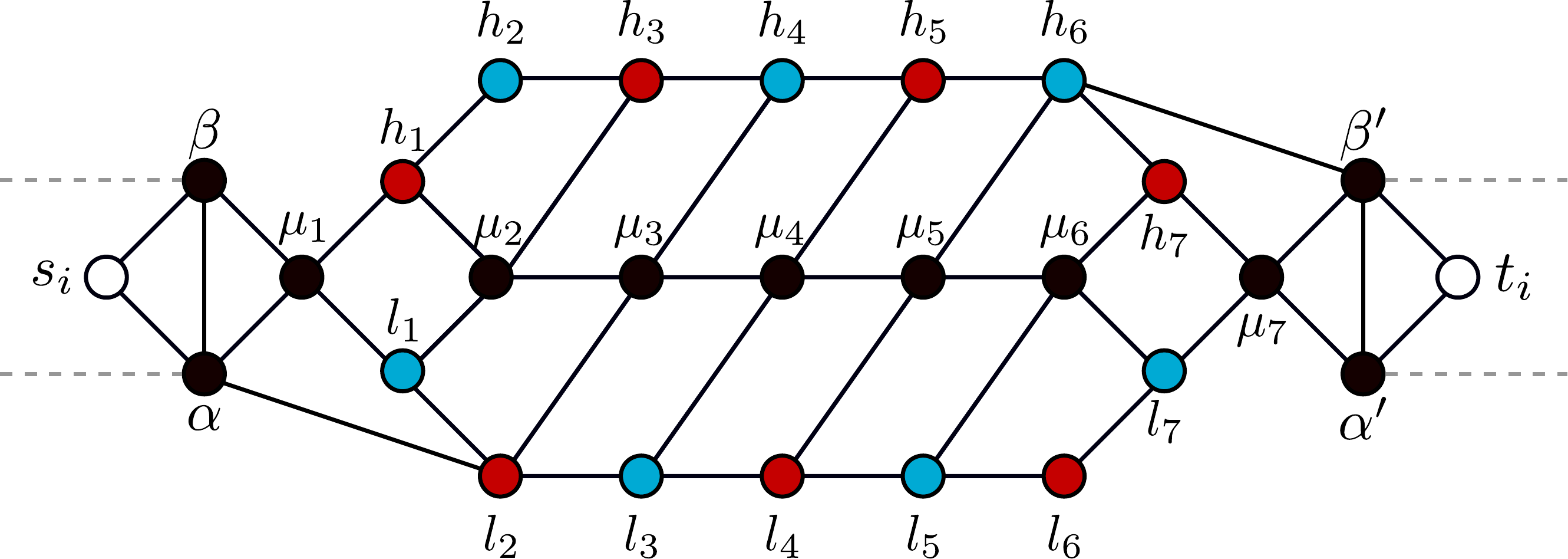}
\caption{Illustration of $x_i$'s gadget, containing $m_i$ vertices for $x_i$ (in blue) and $m_i$ vertices for $\overline{x_i}$ (in red). Vertices colored black require trackers in any minimum tracking set. The dashed edges are added to force trackers in $s_i$ and $t_i$.}
\label{fig:var_gadget}
\end{figure}

There are exactly two minimum tracking sets associated with the variable gadget with source $s_i$ and destination $t_i$. One of them corresponds to a true assignment of $x_i$ and the other one to a false assignment. Both of them require tracking the vertices in $R = \{\alpha, \alpha',\beta,\beta',\mu_1,\mu_{m_i}\}$, as well as the remaining $\mu_k$. In addition, the true assignment tracks the even-indexed $h_k$ and odd-indexed $l_k$, while the false assignment tracks the odd-indexed $h_k$ and even-indexed $l_k$. This requires $2m_i + 4$ trackers in total.

\begin{restatable}{lemma}{lemvargadgetmintrackingset}\label{lem:no_other_tracking_set}
$\star$ The true and false assignments are the only minimum tracking sets.
\end{restatable}

\subparagraph*{Clause gadget.}

Let $C = \left (\ell_a \vee \ell_b \vee \ell_c \right)$ be a clause in $\phi$ with literals corresponding to variables $x_a, x_b, x_c \in \X$. Its gadget, depicted in \cref{fig:clause_gadget}, is a face $F_C$ consisting of:

\begin{itemize}
  \item literal vertex $\alpha$ from $x_a$'s gadget, corresponding to literal $\ell_a$;
  \item adjacent literal vertices $\beta_1,\overline{\beta_1},\beta_2,\overline{\beta_2},\beta_3$  from $x_b$'s gadget, corresponding to literals alternating between $\ell_b$ and $\overline{\ell_b}$;
  \item literal vertex $\gamma$ from $x_c$'s gadget, corresponding to literal $\ell_c$;
  \item edges $(\alpha,\gamma), (\alpha, \beta_1), (\beta_3,\gamma)$ and the edges from $x_b$'s gadget connecting all of the $\beta_k$ and $\overline{\beta_k}.$
\end{itemize}

\begin{figure}[b]
\centering
\includegraphics[width=\textwidth]{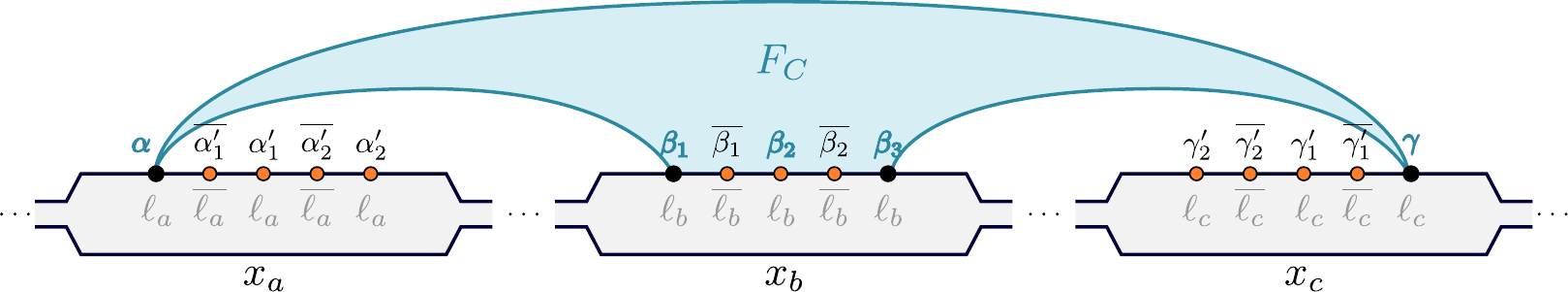}
\caption{Illustration of the gadget for clause $C=(\ell_a \vee \ell_b \vee \ell_c)$, where the vertices in each variable gadget are all adjacent. The entry-exit $(\overline{\beta_1},\overline{\beta_2})$ are responsible for satisfying $C$.}\label{fig:clause_gadget}
\end{figure}

We can increase the lengths of the variable gadgets to any polynomial that provides enough literal vertices for all clauses. Since $\D$ is planar, there are no crossings between clauses. We also impose the following restrictions:
\begin{enumerate}
  \item $\alpha$ cannot be one of $\{h_1,l_{m_a}\}$; this ensures that the only faces in $x_a$'s gadget that do not require 3 trackers do not become untracked. We apply the equivalent restriction to $\beta_1,\beta_3$ and $\gamma$. \label{restriction:on_faces_2trackers}
  \item The $\alpha'_k/\overline{\alpha'_k}$, (see \cref{fig:clause_gadget}) cannot belong to any other clause gadget; these correspond to the 4 literal vertices following $\alpha$ in $x_a$'s gadget and reserving them ensures that non-clause faces, between nested clauses, are tracked. We apply the equivalent restriction to the $\gamma'_k/\overline{\gamma'_k}$. \label{restriction:reserved_vertices}
  \item All literal vertices in a clause need to be on the same side of the spine; this restriction is trivial because $\D$ is rectilinear, but it simplifies the analysis. \label{restriction:same_side_spine}
\end{enumerate}

\begin{restatable}{lemma}{lemclausesatisfiediff}\label{lem:clause_satisfied_iff}
$\star$ Clause $C$ is satisfied if and only if its corresponding gadget face $F_C$ is tracked.
\end{restatable}

\begin{restatable}{theorem}{thmreduction}\label{thm:reduction}
$\star$ There exists a polynomial time reduction from \psat{} to \ptracking{}.
\end{restatable}

\begin{corollary}
  \ptracking{} is NP-hard.
\end{corollary}

It remains to show that \ptracking{} is in NP; Banik et. al. \cite{banik2018polynomial} prove this in the more general case of \tracking{}.

\section{Bounded clique-width graphs}\label{sec:courcelle}

We show that \tracking{} can be solved in linear time when the input graph has bounded clique-width, by applying Courcelle's theorem \cite{courcelle1990monadic,courcelle2012graph,courcelle2000linear}, a powerful meta-theorem that establishes fixed-parameter tractability of \emph{any} graph property that is expressible in monadic second order logic.

\emph{Clique-width}, first introduced by Courcelle et. al. \cite{courcelle1993handle} and revisited by Courcelle and Olariu \cite{courcelle2000upper}, is an important graph parameter that, intuitively, measures the closeness of a graph to a cograph -- a graph with no induced 4-vertex paths. It is closely related to \emph{tree-width}, another influential graph parameter that measures closeness of a graph to a tree and that was first introduced by Bertelé and Brioschi \cite{bertele1972nonserial} and later rediscovered by Halin \cite{halin1976s} and Robertson and Seymour \cite{robertson1986graph}. While both parameters are determined based on specific hierarchical decompositions of a graph, the clique-width is strictly more powerful in the sense that the class of graphs of bounded clique-width includes all graphs of bounded tree-width, but not vice-versa. Details on the relationship between these parameters can be found in Courcelle and Engelfriet \cite{courcelle2012graph}. Graphs of bounded clique-width include series-parallel graphs, outerplanar graphs, pseudoforests, cographs, distance-hereditary graphs, etc.

\subparagraph*{MSO$_1$ vs MSO$_2$.}

Second order logic extends first order logic, by allowing quantification over relations (of any fixed arity) on the elements of the domain of discourse. Monadic second order logic itself only allows quantification over unary relations (subsets of the domain of discourse) and, in the logic of graphs, it comes in two flavors: MSO$_1$ and MSO$_2$. The only distinction between these is that the latter allows edges to be elements of the domain of discourse (and thus be quantified over), while the former does not. Besides the quantifiers ($\forall$ and $\exists$) and the standard logic operations $\neg,\wedge,\vee,\rightarrow$, both logics include predicates for equality $(=)$ and relation membership $(\in)$. In addition, MSO$_1$ includes a predicate $(\sim)$ that determines vertex adjacency and MSO$_2$ includes a predicate for vertex-edge incidence. MSO$_2$ is more expressive, for example: Hamiltonicity can be expressed using MSO$_2$, but not using MSO$_1$. Details on the distinction between the two logics can be found in \cite{courcelle2012graph}.

\subparagraph*{Courcelle's theorem.}

Courcelle et. al. \cite{courcelle1990monadic,courcelle2000linear} showed that any graph property expressed in MSO$_1$ and MSO$_2$ is FPT under clique-width and tree-width (respectively). More specifically, they showed that any MSO$_1$-, MSO$_2$-expressible property can be tested in $f(k,l)\cdot n$ and $g(k',l')\cdot (n + m)$ time (respectively) for graphs of clique-width $k$ and tree-width $k'$, where: $f$ and $g$ are computable functions, $l$ and $l'$ are the lengths of the logic formulas, $n$ is the number of vertices and $m$ is the number of edges. The result for MSO$_2$ is valid in optimization problems with linear evaluation functions \cite{courcelle1993monadic}. Later, Courcelle, Makowsky and Rotics \cite{courcelle2000linear} extended these results for MSO$_1$. Examples of constructing FPT graph algorithms parameterized by clique-width or tree-width, which are based on automata, are given in \cite{DBLP:journals/japll/CourcelleD12,DBLP:journals/tcs/CourcelleD16}. While it is possible to construct tree decompositions of width $k'$ in linear time \cite{bodlaender1996linear}, there is no FPT algorithm for finding clique decompositions of clique-width $k>3$. Fortunately, it is possible to construct a clique decomposition of width exponential in $k$ in cubic time \cite{oum2008approximating}. In this section, we will take advantage of the latter. 

We give an alternative definition for tracking set that is easier to express using the logic of
graphs.

\begin{lemma}[Tracking set]\label{lem:tracking_set_2}
For an undirected graph $G=(V,E)$, a subset $T \subseteq V$ is a tracking set if and only if there is \emph{no} $s-t$ path $P_{st}=P_{ss'}\cup P_{s't'}\cup P_{t't}$ for $s',t' \in V$ and corresponding $s-s'$, $s'-t'$ and $t'-t$ paths, such that:
\begin{enumerate}
  \item There exists an alternative $s'-t'$ path $P'_{s't'} \neq P_{s't'}$ and
  \item $T \cap (P_{s't'} \cup P'_{s't'}) \subseteq \{s',t'\}$.
\end{enumerate}
\end{lemma}

\begin{proof}
  This follows directly from \cref{lem:tracking_set}.
\end{proof}

In the logic formulas presented below, we use lowercase letters to quantify over vertices and uppercase letters to quantify over sets of vertices. We use $s$ and $t$ as free variables and, for convenience, we also use set intersection $(\cap)$, union $(\cup)$ and containment $(\subseteq)$, without explicitly expressing these operations using MSO$_1$.

\begin{mdframed}
\begin{tabularx}{\textwidth}{lll}
\multicolumn{3}{l}{$\textsc{IsTrackingSet}(T,s,t)$} \\
$\iff \nexists P,Q,R \ [\ \exists s',t'\ [$ & \multicolumn{2}{l}{$\textsc{HasPath}(P,s,s')\ \wedge\ \textsc{HasPath}(Q,s',t')\ \wedge\ \textsc{HasPath}(R,t',t)$}\\
& \multicolumn{2}{l}{$\wedge\ P \cap Q = \{s'\}\ \wedge\ Q \cap R = \{t'\}\ \wedge\ P \cap R = \emptyset$}\\
& $\wedge\ \exists Q' \neq Q$ & $[\textsc{HasPath}(Q', s', t')\ \wedge\ T \cap (Q\cup Q') \subseteq \{s',t'\}\ ]]]$\label{eq:mso:t}
\end{tabularx}
\end{mdframed}

The first two lines of the above equivalence establish that $P$, $Q$ and $R$ form an $s-t$ path. The last line restricts $s'$ and $t'$ to be an entry-exit pair with respect to the cycle $Q \cup Q'$ (see \cref{fig:mso}) and, in addition, establishes that the cycle $Q \cup Q'$ is not tracked.

\begin{figure}
\centering
\includegraphics[width=190pt]{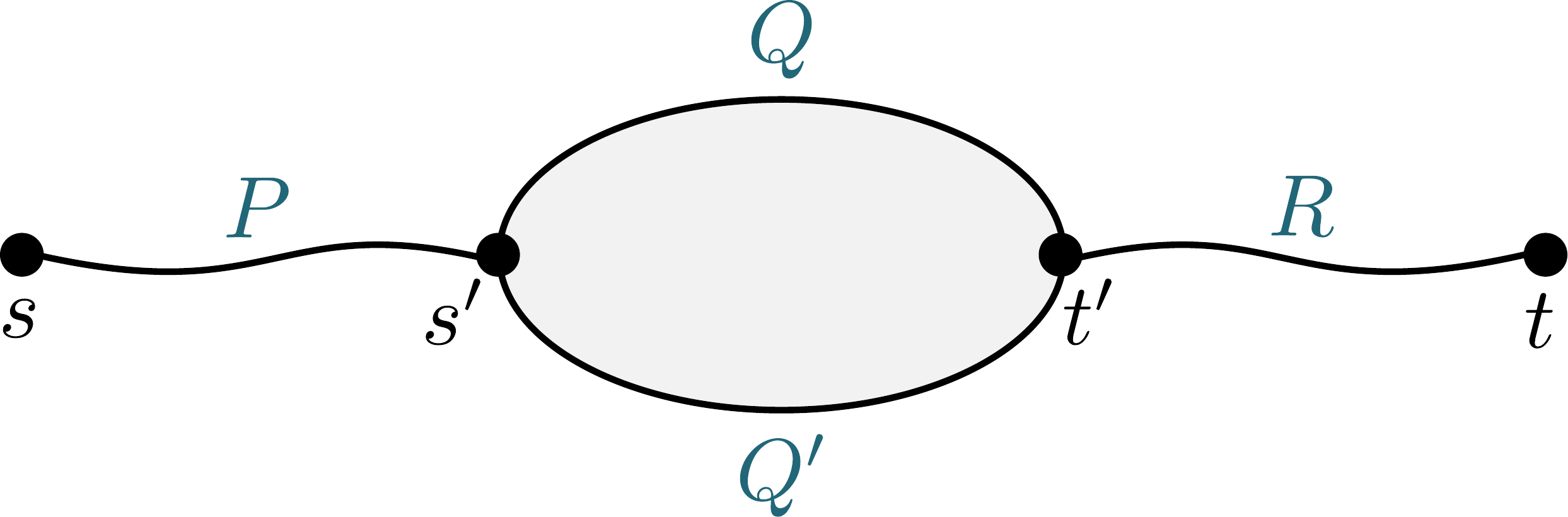}
\caption{Illustration of variable sets $P$, $Q$, $Q'$ and $R$ as well as vertex variables $s$, $s'$, $t'$
and $t$ used in expressing \textsc{IsTrackingSet} using MSO$_1$.}
\label{fig:mso}
\end{figure}

The primitive $\textsc{HasPath}(X, a, b)$, whose input consists of a set $X \subseteq V$ and vertices $a,b \in V$, verifies the existence of a simple path between $a$ and $b$ that only uses vertices in $X$. We define it as follows:

\begin{mdframed}
\begin{tabularx}{\textwidth}{l@{}l}
  $\textsc{HasPath}(X,a,b)$ \\
  $\iff \nexists X_1,X_2 \subseteq X \ [\ $ & $X_1 \cup X_2 = X\ \wedge\ a \in X_1\ \wedge\  b \in X_2\ \wedge\ \neg \left( \exists u\in X_1\ \wedge\ \exists v \in X_2 \ [ u \sim v\ ] \right)]$
\end{tabularx}
\end{mdframed}

\begin{remark}\label{rem:has_path}
\textsc{HasPath}$(X,a,b)$ is correctly expressed under MSO$_1$ and it correctly verifies that there exists an $a-b$ path using only vertices in $X$.
\end{remark}

\begin{remark}\label{rem:is_tracking_set}
\textsc{IsTrackingSet} is correctly expressed under MSO$_1$ and it correctly verifies that the given subset of vertices is a tracking set.
\end{remark}

\begin{theorem}
\tracking{}$(G,s,t)$ can be solved in polynomial time if $G$ has bounded clique-width. Moreover, if a clique decomposition of bounded width is given, it can be solved in linear time.
\end{theorem}

\begin{proof}
This follows directly from \cref{lem:tracking_set_2}, \cref{rem:has_path,rem:is_tracking_set}, and the linear time algorithm given by Courcelle \cite{courcelle2000linear} for any optimization problem on graphs of bounded clique-width, whose decomposition is given in advance.
\end{proof}

\section{Conclusion and open questions}\label{sec:conclusion}

We showed that \ptracking{} is NP-complete and we give a 4-approximation algorithm. We also show that, for graphs of bounded-clique width, \tracking{} can be solved in linear time by applying Courcelle's theorem, as long as its clique decomposition is given in advance. A natural direction of future study would be to improve the approximation ratio of \ptracking{} or establish constant approximation factors for graphs of larger genus or, more generally, for \tracking{}. Another open question is to establish the difficulty of \tracking{} on directed graphs: on one side planar directed acyclic graphs are not known to be NP-hard and, on the other side, it is not known whether its general or planar versions are in NP. Finally, it would be interesting to find efficient algorithms for graphs of bounded tree-width or clique-width without resorting to the finite automaton approach used in Courcelle's theorem.

\bibliographystyle{plainurl}
\bibliography{tracking_planar}

\pagebreak

\appendix

\begin{figure}
\centering
\begin{subfigure}[b]{0.3\textwidth}\label{subfig:add_faceA}
  \includegraphics[height=100pt]{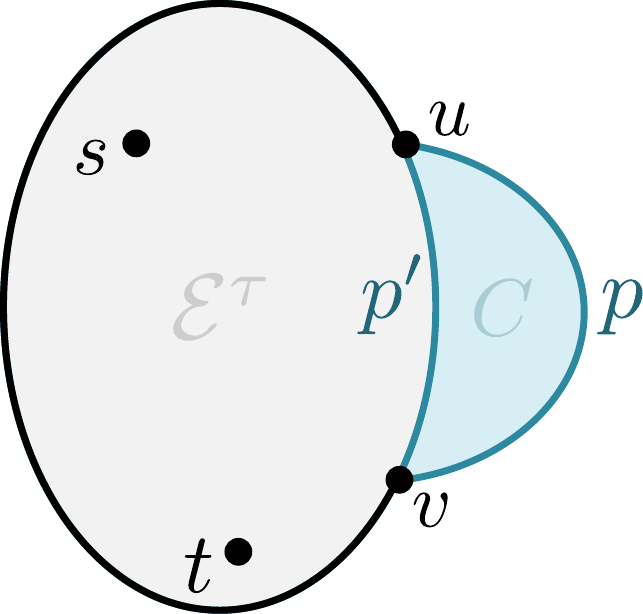}
  \caption{Valid}
\end{subfigure}
~
\begin{subfigure}[b]{0.3\textwidth}\label{subfig:add_faceB}
  \includegraphics[height=100pt]{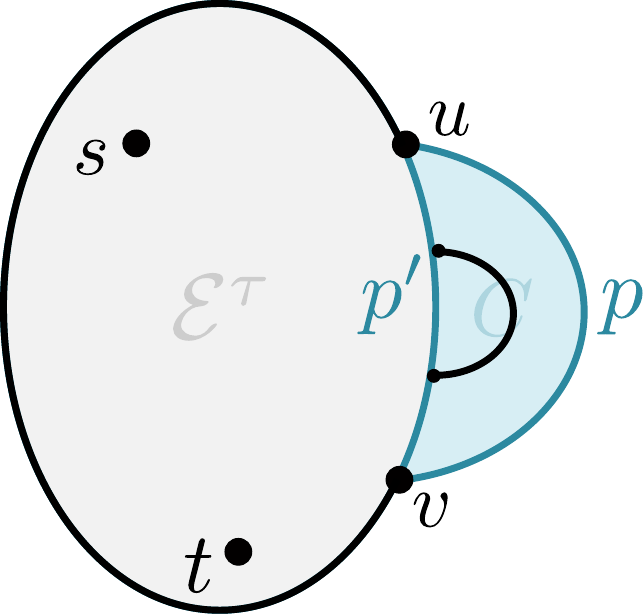}
  \caption{Invalid}
\end{subfigure}
~
\begin{subfigure}[b]{0.3\textwidth}\label{subfig:add_faceC}
  \includegraphics[height=100pt]{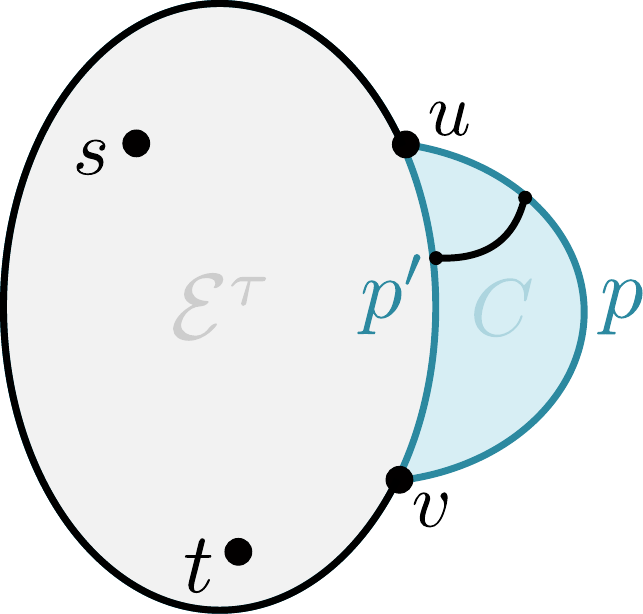}
  \caption{Invalid}
\end{subfigure}
\caption{Illustration of valid and invalid choices of vertices $u$ and $v$ in the proof of \cref
{lem:OPT_bound}. Cycle $C$ must correspond to a face in the fixed embedding $\E$.}\label{fig:add_face}
\end{figure}

\section{Deferred proofs on the approximation algorithm}\label{sec:missing_proofs_approx}

\begin{figure}
\centering
\begin{subfigure}[c]{.45\textwidth}
  \begin{subfigure}[t]{\textwidth}
  \centering
    \includegraphics[width=178pt]{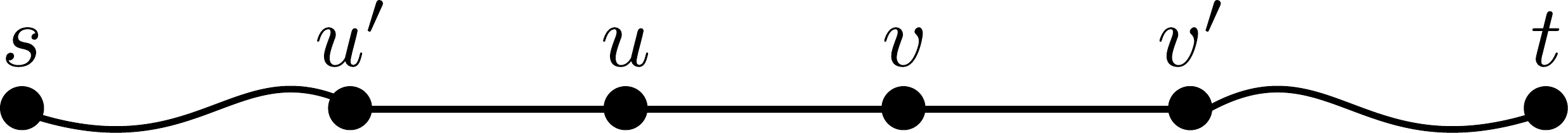}
  \caption{Illustration of \hyperref[reduction:2]{Reduction~2}, where $deg(u) = deg(v) = 2$.}
  \label{fig:reduction2_copy}
  \end{subfigure}

  \vspace{1em}

  \begin{subfigure}[t]{\textwidth}
  \centering
    \includegraphics[width=130pt]{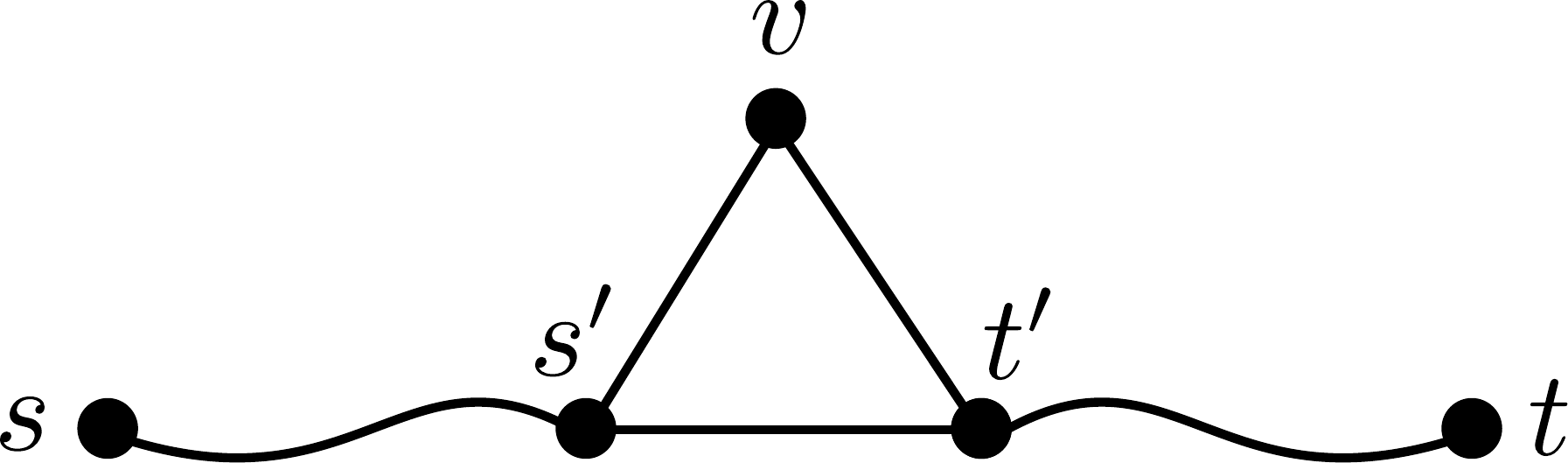}
  \caption{Illustration of \hyperref[reduction:3]{Reduction~3}, where $deg(s')\ge 3$, $deg(t')\ge 3$ and $deg(v)=2$.}
  \label{fig:reduction3_copy}
  \end{subfigure}
\end{subfigure}
\hspace{2em}
\begin{subfigure}[c]{.45\textwidth}
\centering
\includegraphics[height=100pt]{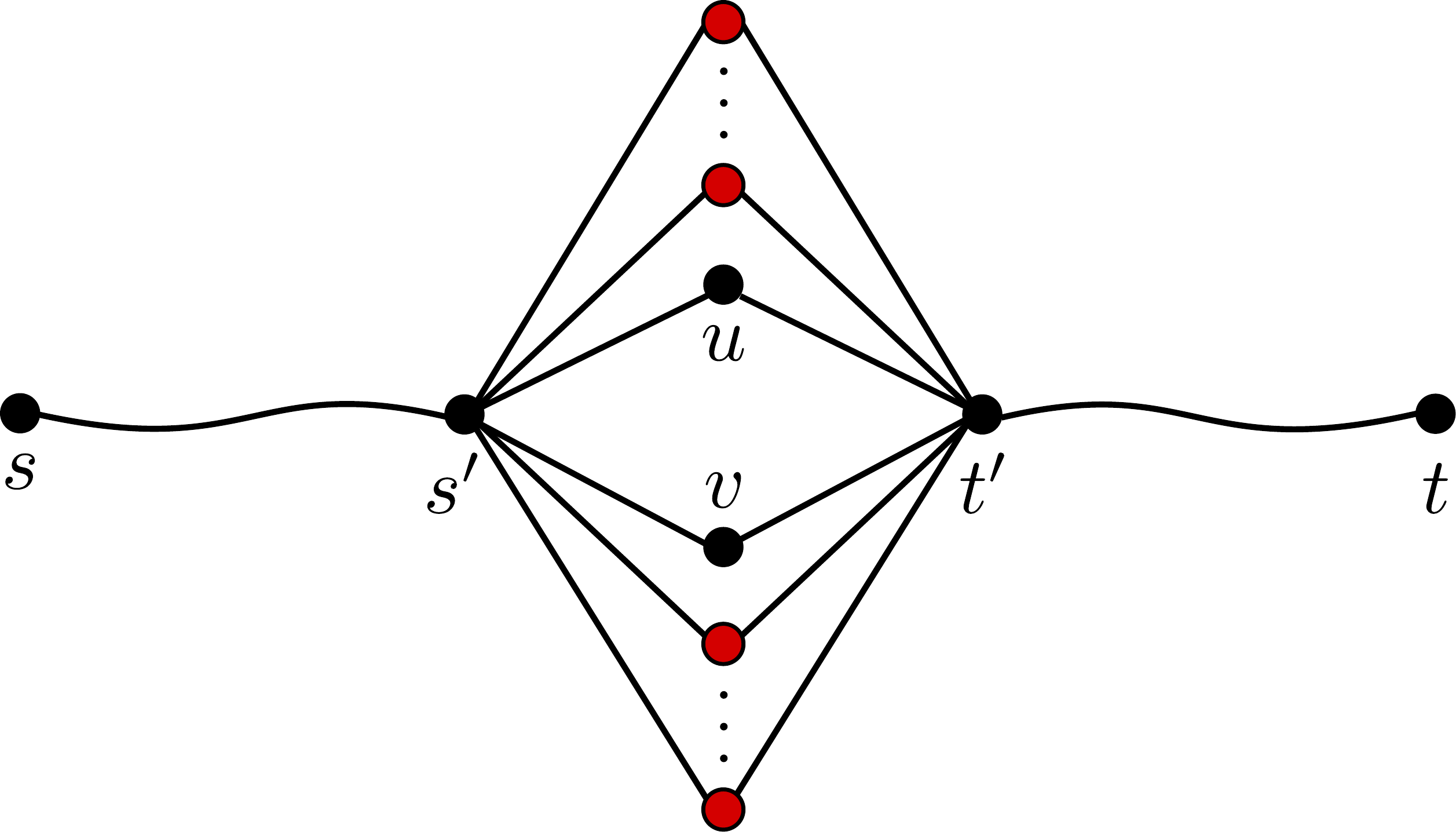}
\caption{Illustration of \hyperref[reduction:4]{Reduction~4}, where $deg(s')\ge 3$, $deg(t')\ge 3$ and $deg(u)=deg(v)=2$.}
\label{fig:reduction4_copy}
\end{subfigure}
\caption{Illustration of \hyperref[reduction:2]{Reductions~2}, \hyperref[reduction:3]{3} and \hyperref[reduction:4]{4}.}
\end{figure}

\FloatBarrier

\lemrone*
\begin{proof}
See proof in \cite{banik2018polynomial}.
\end{proof}

\lemrtwo*
\begin{proof}
Banik et. al. \cite[Lemma 16]{banik2018polynomial} showed that having three adjacent vertices of degree 2 is redundant. Here, we extend their proof and show that we can also get rid of the second adjacent vertex of degree 2. Let $u$ and $v$ be the two adjacent vertices of degree 2, and let $u'$ be the other neighbor of $u$.
Notice that $v,u'$ are not adjacent, otherwise the $u,v$ edge would then not belong to an $s-t$ path and thus be removed by \hyperref[reduction:1]{Reduction~1}. Hence, no parallel edges are added to the graph after this reduction. We argue that this reduction is safe by showing that there exists a minimum tracking set that does not include $u$ and $v$ simultaneously. Take any minimum tracking set $T$ that includes $u$ and $v$. We can always move the tracker from $u$ to $u'$; this remains a tracking set, because $u$ immediately follows or precedes $u'$ on any $s-t$ path.
This can only decrease $T$'s cardinality. This reduction can be done in $O(poly(n))$ time, by repeatedly checking for the existence of adjacent vertices of degree 2.
\end{proof}

\lemrthree*
\begin{proof}
Let $v \notin \{s,t\}$ be the vertex of degree 2 in the triangle $\Delta vs't'$ (see \cref{fig:reduction3_copy}). Then, it must be the case that $s'$ and $t'$ form an entry-exit pair. This follows from (i) the fact that there exists an entry-exit pair in $\Delta vs't'$ (by \cref{lem:cycle_has_entry_exit}) and (ii) the fact that $v$ cannot be in an entry-exit pair (by \cref{clm:deg2_cant_entry_exit}). Then, by \cref{lem:tracking_set}, any feasible solution must place a tracker on $v$. Therefore, $v$ and its edges can be removed, since $v$ is neither a cut-vertex, nor in any entry-exit pair, so that its removal does not eliminate an untracked cycle. Clearly, this reduction can be done in $O(poly(n))$ time.
\end{proof}

\lemrfour*
\begin{proof}
Let $u,v \notin \{s,t\}$ be the two vertices of degree 2 that connect the same pair of vertices $s'$ and $t'$. Similarly to the proof of \cref{lem:reduction3}, $s'$ and $t'$ must form an entry-exit pair. So, by \cref{lem:tracking_set}, any feasible solution must place a tracker in either $u$ or $v$. By symmetry, we can track and remove $u$ and its edges. Therefore, \hyperref[reduction:4]{Reduction~4} is safe by \cref{clm:deg2_cant_entry_exit} and the facts that $u$ is neither a cut-vertex nor in an entry-exit pair.
\end{proof}

\section{Deferred proofs on NP-completeness}\label{sec:missing_proofs_hardness}

To prove \cref{lem:no_other_tracking_set}, we first prove a couple of helpful lemmas and make a few observations:

\begin{itemize}
  \item Each vertex in $R$ belongs to a triangle, such that the other two vertices form an entry-exit pair, so must be tracked.
  \item Each square face, besides the two including $h_2$ and $l_{m_i-1}$, requires $3$ tracked vertices. 
  \item Two adjacent vertices cannot both be untracked.
\end{itemize}

\begin{lemma}\label{lem:var_gadget_is_trackingset}
The true/false assignments of $x_i$ correspond to tracking sets with respect to $x_i$'s gadget with source $s_i$ and destination $t_i$.
\end{lemma}

\begin{proof}
In a true/false assignment of $x_i$, every face of the gadget is trivially tracked, except the faces including $s_i/t_i$ (which are clearly tracked) and the faces including $h_2$ and $l_{m_i-1}$. Though the latter faces may only contain 2 trackers (depending on $x_i$'s truth value and/or $m_i$'s parity), they are nevertheless tracked because one of these trackers cannot be in an entry-exit pair. The remaining cycles, i.e. the ones which are not faces, are also trivially tracked: since these cycles must have size at least 6, they must contain 3 trackers given the observation that no adjacent vertices are untracked.
\end{proof}

\begin{lemma}\label{lem:3_trackers_column}
In a minimum tracking set, each column has exactly 2 trackers.
\end{lemma}

\begin{proof}
Since the true/false assignments achieve this property, we only have to show that each column requires at least two trackers. Assume that a minimum tracking set has only one tracker in column $1<k<m_i$; it must be on $\mu_k$. Then all vertices on columns $k-1,k+1$ must be tracked. For the average number of trackers per column to be at most $2$, the number of trackers per column must be an alternating sequence of $\dots 3,1,3,1, \dots$, with column $1$ and/or column $m_i$ only having 1 tracker, which contradicts the square face property.
\end{proof}

\lemvargadgetmintrackingset*
\begin{proof}
Assume that some $\mu_k$ is not tracked in a minimum tracking set, for $1<k<m_i$. By \cref{lem:3_trackers_column}, we only have to show that this causes a column to have 3 trackers. If $k=m_i-1$ then column $m_i$ must have three trackers. Otherwise, $\mu_{k+1},h_{k+1},h_{k+2}$ must be tracked, since they share a square face with $\mu_k$. If $l_{k+1}$ is tracked we are done, otherwise, $l_{k+2}$ must be tracked. Additionally, $\mu_{k+2}$ must be tracked, either by the square face property, or in the case where $k=m_i-2$, because it is in $R$. Then, column $k+2$ has 3 trackers.

Now, if all the $\mu_k$ are tracked, then the only minimum tracking sets are the true and false assignments, by the observation that two adjacent vertices cannot both be untracked and \cref{lem:var_gadget_is_trackingset}.
\end{proof}

\lemclausesatisfiediff*
\begin{proof}
The entry-exit $(\overline{\beta_1}, \overline{\beta_2})$ is the only one, with respect to $F_C$, that is tracked if and only if $C$ is satisfied. The remaining entry-exit pairs are tracked by either a $\beta_k$ or a $\overline{\beta_k}$.
\end{proof}

\thmreduction*
\begin{proof}
Let $(G,\G)$ be an instance of \ptracking{} that results from applying the transformation described above to an instance $(\phi,\D)$ of \psat{}, where $\G$ and $\D$ are the underlying planar embeddings. We show that $\phi$ is satisfiable if and only if $G$ has a tracking set of size $T=\left(\sum_{i=1}^p2m_i + 4\right) + p-1$.

$(\Leftarrow)$ Choose the truth assignment of each variable according to the given tracking set. The implication follows from \cref{lem:no_other_tracking_set} and the fact that if some clause in $\phi$ was not satisfied, then its gadget face would have been untracked, a contradiction.

$(\Rightarrow)$ Place $T$ trackers on the literal vertices that correspond to the satisfiable truth assignment of every variable and on all non-literal vertices (except $s$ and $t$). We show that this corresponds to a tracking set by arguing that every cycle $C$ in $G$ is tracked. We distinguish between two cases:

\begin{enumerate}[{Case} 1:]
\item $C$ contains no clause edges.

Then $C$ is tracked by (almost) the same argument given in \cref{lem:var_gadget_is_trackingset} that shows that a truth assignment of $x_i$ corresponds to a tracking set with respect to $x_i$'s gadget. Notice that, because of \hyperref[restriction:on_faces_2trackers]{Restriction~1}, clause edges are only added to faces that require 3 trackers, so this does not change the argument for the faces which do not require 3 trackers. The only differences are: (i) the addition of the edges that force trackers in every $s_i$, which only helps the argument, and (ii) the fact that $C$ may span multiple variable gadgets, in which case $C$ must traverse at least 3 trackers on the non-literal vertices between two variable gadgets.

\item $C$ contains clause edges.

Notice that, by construction of $G$, $C$ must have at least one spine edge. If $C$ corresponds to a clause face, then it must be tracked by \cref{lem:clause_satisfied_iff}. Otherwise, we show that $C$ contains at least 3 trackers and, thus, is trivially tracked. Let us think of $C$ as alternating non-empty paths of two types: \emph{clause paths}, which only contain clause edges and \emph{spine paths}, which only contain spine edges. To avoid dealing with complex cycles, we observe that each spine path in $C$ must contain at least 1 tracker; this follows from the fact at least one of the 2 vertices sharing a spine edge must have a tracker (see \hyperref[par:var_gadget]{variable gadget}). Thus, let us assume that $C$ contains no more than 2 spine/clause paths, or otherwise $C$ immediately contains 3 trackers. Notice that if one of the spine paths spans 2 or more variable gadgets, then it must traverse at least 3 trackers on the non-literal vertices between two variable gadgets. Since every clause in $\phi$ contains exactly 3 distinct literals and $C$ is simple, the only cases where none of the spine paths span multiple variable gadgets are the following:

\begin{enumerate}[(i)]
  \item $C$ contains exactly 2 clause paths in different sides of the spine.

  Then, the 2 spine paths connecting the two sides of the spine must traverse 2 trackers each (see \hyperref[par:var_gadget]{variable gadget}). Therefore, $C$ contains at least 4 trackers.

  \item $C$ contains exactly 2 clause paths on the same side of the spine, which both start or both end at the same variable gadget, one nested in the other.

  Then, by \hyperref[restriction:reserved_vertices]{Restriction~2} one of the spine paths contains at least 6 vertices, half of which must be tracked.
\end{enumerate}
\end{enumerate}
\end{proof}

\begin{figure}
\begin{subfigure}[t]{\textwidth}
\centering
\includegraphics[width=120pt]{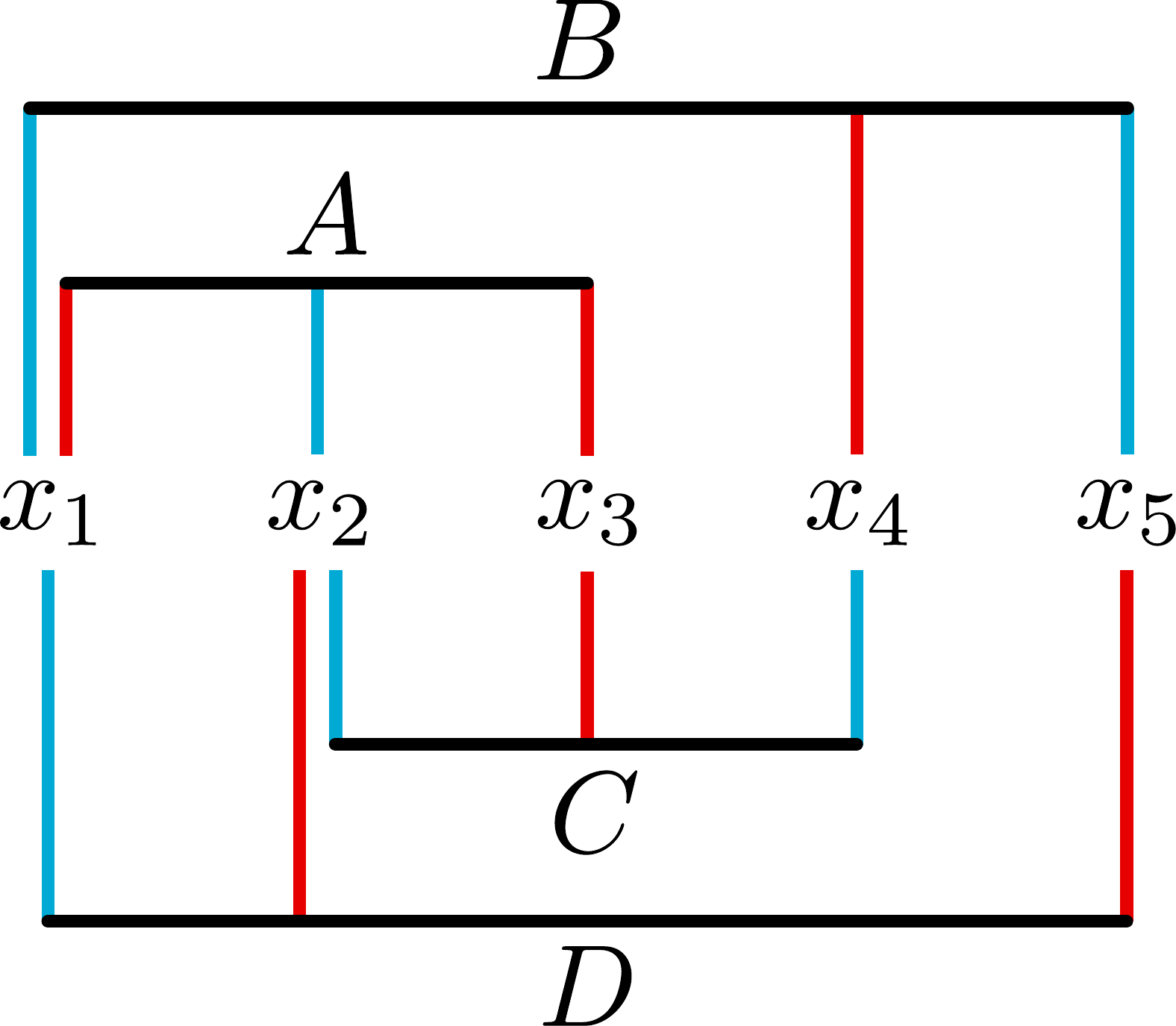}
\caption{Example of a \psat{} instance drawn in a rectilinear fashion with clauses
${A = \left (\overline{x_1} \vee x_2 \vee \overline{x_3} \right)}$,
${B = \left (x_1 \vee \overline {x_4} \vee x_5 \right)}$,
${C = \left (x_2 \vee \overline{x_3} \vee x_4 \right)}$,
${D = \left (x_1 \vee \overline{x_2} \vee \overline {x_5} \right)}$.}
\label{subfig:rect_sat}
\end{subfigure}

\vspace{2em}

\begin{subfigure}[t]{\textwidth}
\centering
\includegraphics[width=\textwidth]{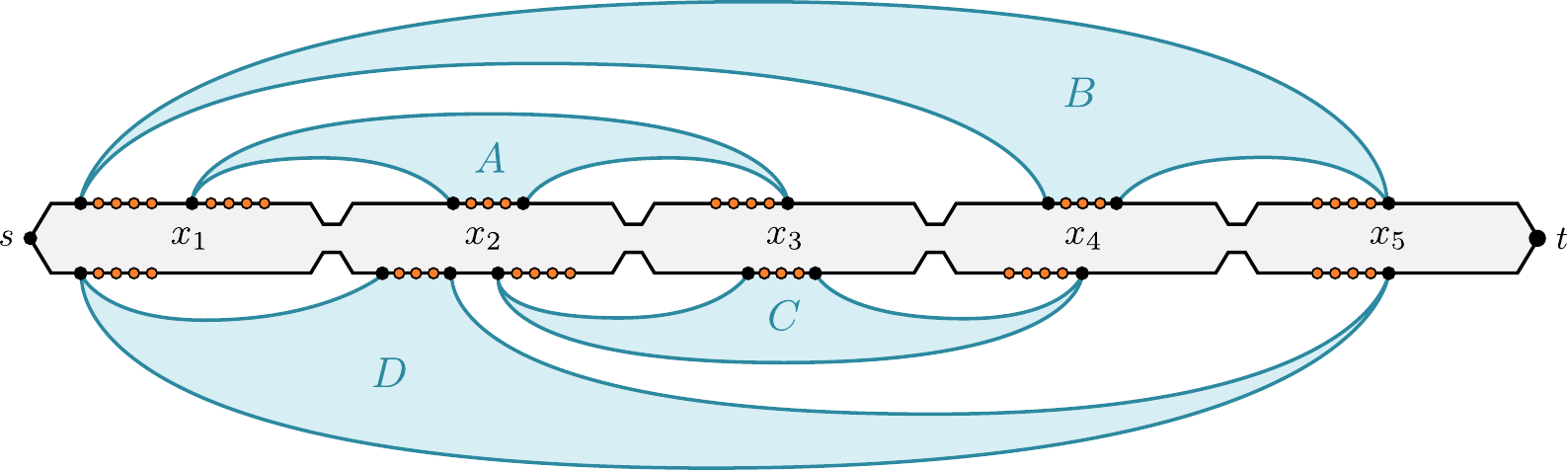}
\caption{Example of a \ptracking{} instance. The orange vertices ensure that no cycle is untracked (see \cref{fig:clause_gadget}).}
\label{subfig:reduction_overview}
\end{subfigure}
\caption{Illustration of a \psat{} instance (above) and the corresponding \ptracking{} instance associated with the reduction described in \cref{sec:hardness} (below).}\label{fig:reduction_overview}
\end{figure}

\end{document}